\documentclass[10pt,a4paper,foot]{amsart}
\usepackage{amsmath,amsthm,amstext,amssymb,a4,stmaryrd,mathrsfs}
\usepackage{graphicx}
\usepackage{paralist}
\usepackage[T1]{fontenc}

\newcommand{\R}{\ensuremath{\mathbb{R}}}

\theoremstyle{plain}

\newtheorem{proposition}{Proposition}[]

\theoremstyle{remark}

\theoremstyle{definition}

\DeclareMathOperator \grad {grad}
\DeclareMathOperator \ric {Ric}

\DeclareMathOperator \interior {int}

\DeclareMathOperator \isgn {isgn}

\newcommand{\pdef}{\mathrel{\mathop:}=}

\begin{document}

\title[]{Geometric Analysis of Ori-type Spacetimes}
\author{J. Dietz}

\author{A. Dirmeier$^2$ and M. Scherfner$^2$}
\address{$^2$Department of Mathematics, Technische Universit{\"a}t Berlin,
Str.~d.~17.~Juni 136, 10623 Berlin, Germany}

\email{scherfner@math.tu-berlin.de}

\begin{abstract}
In 1993 A.~Ori \cite{Ori93} presented spacetimes violating the chronology condition in order to answer the question whether a time machine construction has to violate the weak energy condition or not. Later, in 2005 \cite{Ori05}, he constructed a class of time machine solutions with compact vacuum core. Both classes include an interesting global structure and it is possible to obtain closed timelike curves. Besides we focus on the geometric structure, in particular symmetries and geodesics, if feasible, and visualize several aspects.
\end{abstract}

\maketitle

\section{Introduction} 

The spacetimes analyzed here, known as Ori-type spacetimes,  were foremost described in \cite{Ori93} and \cite{Ori05}. The first paper presents a time machine model in which closed timelike curves (CTCs) evolve in a bounded region of
space from a well-behaved spacelike initial slice. This slice $S$, just as the entire spacetime, is asymptotically
flat and topologically trivial. In addition, this model fulfills the weak energy condition on $S$
and up until and beyond the time slice, displaying the causality violation.

The second paper uses particular vacuum solutions to construct time machine models where the causality violation occurs inside an empty torus, constituting the core of the time machine. In that case the matter field surrounding the (empty) torus satisfies the weak, dominant, and strong energy conditions.

Here we will comprehensively discuss the geometric and CTC structure of the class of Ricci flat chronology violating spacetimes given in \cite{Ori05}. The analysis for the other class discussed in \cite{Ori93} will only be given as a supplement, since---although it is the predecessor in some sense---the details are much harder to reveal because of the more complicated metric tensor, so here the research is still going on. 

The structure of other particular Ori-type spacetimes---pseudo Schwarzschild and pseudo Kerr---are discussed in \cite{DieDirSch}.

Before starting with our main investigations we describe the important underlying time machine structure.

\section{Preliminaries}

Particles move through spacetime along causal curves, null curves are the trajectories of massless particles whereas particles with non-vanishing mass are described by timelike curves. Suppose now that we are given a spacetime $M$ containing CTCs. The physical interpretation would be that a particle moving once around a CTC would encounter itself in its own past. A similar argument can be adduced in the case of null curves. This is the reason why spacetimes with closed causal curves (CCCs) have been qualified as time machine spacetimes.

While it is easy to construct spacetimes containing CTCs or CCCs, it is more difficult to find examples of spacetimes that are---to some extent---physically reasonable. In \cite{Ori07}, A.~Ori presents a list, setting constraints on what should be classified as a realistic time machine model.

One essential notion in General Relativity is the Cauchy development or domain of dependence of a subset $\mathcal A$ of a spacetime $M$. It is connected with the problem of formulating the Einstein equations as a well posed initial value problem and the concept of global hyperbolicity.

From a realistic point of view it is desirable to obtain CCCs as the unevitable consequence of a Cauchy development. However, the interior of the Cauchy development of an achronal subset $\mathcal A$ is globally hyperbolic and, therefore, excludes any CCCs. The most one can hope for is that there are points in the boundary of the Cauchy development which lie on CCCs. This is the central idea of a time machine structure (TM-structure) for a spacetime.

We will consider $M$ to be a four-dimensional manifold and a metric tensor $g$ on $M$ with signature $(-+++)$, such that the Lorentzian manifold $(M,g)$ together with a fixed time-orientation is a spacetime. Let $\mathcal V(M)$ denote the set of all points $p \in M$ that lie on some closed causal curve in $M$. We shall refer to $\mathcal V(M)$ as the causality violating region or the time machine. If $p \in \mathcal V(M)$, we say that causality is violated at $p$. The causality condition is said to hold on a subset $\mathcal U \subset M$ if $\mathcal V(M) \cap \mathcal U = \emptyset$.

\label{def-TM-structure}
 We say that $(M,g)$ has TM-structure if the following three properties hold for $M$:
\begin{enumerate} [(TM1)]
 \item There exists an open subset $\mathcal U$ of $M$ on which the causality condition holds,
 \item there is a spacelike hypersurface $\mathscr S$ contained in $\mathcal U$, such that
 \item for an achronal, compact subset $\mathscr C \subset \mathscr S$ the future domain of dependence of $\mathscr C$ contains points in its boundary at which causality is violated, i.e., $\overline{D^+(\mathscr C)} \cap \mathcal V(M) \neq \emptyset$. In this case the causality violating region $\mathcal V(M)$ is said to be \textit{compactly constructed}.
\end{enumerate}

This seems to be a reasonable condition for any spacetime that may be viewed as a time machine model. Although there definitely are many more additional conditions one could impose, we will focus on TM-structure here. It should be noted that models with TM-structure can satisfy the weak, strong and dominant energy condition (see \cite{HawEllis}) for the matter content of the spacetime if properly constructed.  

{\em Example:} Consider the smooth manifold $ M \pdef \mathbb{S}^1 \times \R^3$ with the global coordinate system $(\phi,x,y,z)$ on $M$, $\phi$ being a circular coordinate. Define the metric tensor on $M$ by
\begin{equation}
 g \pdef -d\phi \otimes d\phi + dx \otimes dx + dy \otimes dy + dz \otimes dz.
\end{equation}
Then $(M,g)$ becomes a vacuum spacetime, the so-called (four-dimensional) Lorentz cylinder if we require the timelike vector field $\partial_\phi$ to be future pointing. We see immediately that all $\phi$-coordinate lines are CTCs. Hence, every point of $M$ lies on a closed timelike curve.

It can be shown that $\mathcal V(M) = \emptyset$ for any globally hyperbolic spacetime $M$ (cf. \cite{ONeill}, Chapter 14). On the other hand, in the case of the Lorentz cylinder, we have $\mathcal V(M) = M$. Mainly because of physical reasons (or everyday experience) we do not want to consider the Lorentz cylinder as a spacetime with TM-structure.


The idea behind a TM-structure is the following: One can think of $D^+(\mathscr C)$ as the set of points $p \in M$ that are predictable from $\mathscr C$ in the sense that no inextendible causal curve through $q \in D^+(\mathscr C)$ can avoid $\mathscr C$. Because the set $\interior D^+(\mathscr C)$ does not contain points at which causality is violated (for it is a subset of $D^+(\mathscr C)$ on which the so-called strong causality condition holds), the most one can hope for is that there are points in the boundary of $D^+(\mathscr C)$ contained in $\mathcal V(M)$. 

As it can be seen from the definition of the TM-structure, spacelike hypersurfaces are an important entity in the context of TM-structures. Therefore, we include the following criterion.
\begin{proposition}
 Let $f \colon M \to \R$ be a smooth function on a Lorentzian manifold $M$ and $a \in \R$ be in the image of $f$. Then the sets
\begin{gather}
 H_t \pdef \{ p \in f^{-1}(a) \colon \langle (\grad f)_p, (\grad f)_p \rangle >0 \}	\\
\intertext{and}	
 H_s \pdef \{ p \in f^{-1}(a) \colon \langle (\grad f)_p, (\grad f)_p \rangle <0 \}
\end{gather}
are hypersurfaces in $M$ and $H_t$ is timelike, $H_s$ spacelike.
\end{proposition}
\begin{proof}
 The set $U \pdef \{ p \in M \colon \langle (\grad f)_p, (\grad f)_p \rangle >0 \} \subset M$ is open, hence a Lorentzian manifold in the usual manner. Because of the defining property of $U$ the one form $df|_U$ is nowhere vanishing. Therefore, $H_t$ is a hypersurface in $U$, hence in $M$. If $ p \in H_t$ and $u \in T_pH_t$, we have
\begin{equation}
 g( (\grad f)_p , u) = u(f) = u(f|_{H_t}) =0
\end{equation}
since $f$ is constant on  $H_t$. This proves the decomposition
\begin{equation}
 T_pM = T_pH_t \oplus \R(\grad f)_p
\end{equation}
and because $(\grad f)_p$ is spacelike this implies that $g|_{T_pH_t}$ is of index $1$, i.e., $H_t$ is timelike. 
If we consider $H_s$, almost the same arguments apply, except that now $(\grad f)_p$ is timelike. Hence $T_pH_s$ is a spacelike subspace of $T_pM$, i.e, $H_s$ is spacelike.
\end{proof}
In coordinates $x^1,\ldots,x^n$ the differential $df$ of $f$ is $df = \frac{\partial f}{\partial x^i} dx^i$ and
\begin{equation}
 \langle \grad f, \grad f \rangle = g^{ij} \frac{\partial f}{\partial x^i} \frac{\partial f}{\partial x^j}.
\end{equation}
Using this formula and applying the preceding proposition to the coordinate neighborhood, one obtains:
\begin{proposition} \label{prop-hypersurfaces-in-coord}
 If $x^1,\ldots,x^n$ are local coordinates of a Lorentzian manifold $M$, the coordinate slices $x^k = const$ for $1 \leq k \leq n$ are 
\begin{enumerate}
 \item spacelike hypersurfaces if $g^{kk} <0$,
 \item timelike hypersurfaces if $g^{kk} >0$
\end{enumerate}
for all points $p$ in the coordinate neighborhood.
\end{proposition}

\section{A Class of Ricci Flat Chronology Violating Spacetimes} \label{chap-Ori05}

Here we investigate spacetimes described in \cite{Ori05} and we start with the smooth manifold $M \pdef \R^3 \times \mathbb{S}^1$ and introduce coordinates $(T,x,y,\phi)$ on $M$, where $T$, $x$ and $y$ are natural coordinates on $\R^3$ and $\phi$ is a circular coordinate on $\mathbb{S}^1$. For any smooth function $f \colon \R^2 \times \mathbb{S}^1 \to \R$ we define the metric tensor $g$ on $M$ by
\begin{equation}
 g \pdef - dT \otimes_s d\phi + dx \otimes dx + dy \otimes dy + ( f(x,y,\phi) -T) d\phi \otimes d\phi,
\end{equation}
with $\otimes_s$ being the symmetrized tensor product. Emphasizing the role of $f$, we write $M_f$ for the semi-Riemannian manifold $(M,g)$. The periodicity of $f$ in its third argument guarantees that $g$ is well defined on all of $M_f$. Since the matrix of the component functions of $g$ has $\det [g_{ij}] =-1$, the metric tensor is non-degenerate and has Lorentzian signature.


The Ricci tensor of $M_f$ is given by
\begin{equation}
 \ric = - \frac{1}{2} \left( \frac{\partial^2 f}{\partial x^2} + \frac{\partial^2 f}{\partial y^2} \right) d\phi \otimes d\phi,
\end{equation}
so we assume that $f$ satisfies the relation

\begin{equation}
 f_{xx} + f_{yy} =0
\end{equation}
to get Ricci flat spacetimes $M_f$. As the timelike unit vector field
\begin{equation}
 \partial_\phi + \frac{1}{2}(f-T +1) \partial_T
\end{equation}
shows, $M_f$ is time orientable.

\subsection{General Properties}
Before we make an explicit choice for $f$ and embark on a more detailed study of a special example of the described class, we shall consider the general case.

\subsubsection{Closed Timelike Curves in $M_f$}
Of course, it is the topological factor $\mathbb{S}^1$ in $M_f$ and the metric component $g_{\phi \phi} = \nolinebreak f - T$ that is responsible for the appearance of CTCs. 
Given $T$, $x$ and $y$ the curves 
\begin{equation} \label{CTCs_general}
 [-\pi,\pi] \to M, \qquad s \mapsto \gamma_{T,x,y}(s) \pdef (T,x,y,s)
\end{equation}
are closed and timelike provided
\begin{equation}
 f(x,y,s) - T < 0 \qquad \text{for all} \quad s \in [-\pi,\pi].
\end{equation}
Depending on the explicit form of $f$, this is a condition on the coordinates $(T,x,y)$ determining a non-empty subset of $\R^3$ that consists of points sitting on CTCs of the form \eqref{CTCs_general}. 
By Prop.~\ref{prop-hypersurfaces-in-coord}, the hypersurfaces $H_T$ defined by $T=const$ are 

\begin{equation}
 \left.
  \begin{aligned}
  \text{timelike} \\
  \text{spacelike}
 \end{aligned} \right\}
 \qquad \text{if} \quad g^{TT} = T - f \quad \text{is} \qquad
 \begin{cases}
 \; > 0, \\
 \; < 0,
 \end{cases}
\end{equation}
respectively. Therefore, any timelike curve in $H_T$, in particular CTCs, must be contained in that part of $H_T$ where $T - f > 0$. The same result may also be obtained by calculating $g(\gamma',\gamma')$ directly with the condition $dT(\gamma') = 0$. We will see later for one choice of $f$ how these two ingredients---the region of CTCs and the causal character of $H_T$---allow for a TM-structure of $M_f$.

\subsubsection{Killing Vector Fields of $M_f$}\label{sec1}
Let $\mathcal{L}$ be the Lie derivative, such that the Killing equation for a vector field $X$ is given by $\mathcal{L}_Xg=0$. Furthermore we denote the algebra of Killing vector fields on some open subset $U\subset M$ of a manifold $M$ by $i(U)$. We will make repeated use of the following proposition, the proof of which can be found in \cite{Petrov}, p.~61. Here $R$ denotes the curvature tensor of $M$.

\begin{proposition}\label{prop-dimension-Killing}
 If $X \in i(M)$ is a Killing field, then
\[
 \mathcal L_X D^{(k)} R = 0 \qquad \text{for all} \quad k \in \mathbb N,
\]
i.e., the Lie derivative in the direction of $X$ of any covariant derivative of the Riemann curvature tensor vanishes.  Let $L_k$ denote the linear system $(\mathcal L_X D^{(k)} R)_p$ for some $p \in M$ and $k \in \mathbb N$. Then
\[
 \dim i(M) \leq \dim \ker L_k. 
\]
\end{proposition}

For arbitrary $f$, depending on all three coordinates $x$, $y$ and $\phi$, there is no reason to expect an abundance of non-trivial, linearly independent Killing vector fields on $M_f$. There is only one obvious and non-tivial Killing vector field, which is actually not even globally defined.

\begin{proposition} \label{killgen}
 For arbitrary $f$ there is only one local Killing vector field $K$ on $M_f$. In coordinates,
 \begin{equation}
  K = e^{-\frac{\phi}{2}} \partial_T.
 \end{equation}
\end{proposition}

\begin{proof}
 First of all, note that $K$ cannot be extended to all of $M_f$. Formally not entirely correct, we may say that this is due to the fact that the component function $e^{-\frac{\phi}{2}}$ of $K$ is not properly periodic in $\phi$. If $\phi$ takes values in $(-\pi,\pi)$ the problem arises at points $p \in M_f$ which are not covered by the coordinate system $(T,x,y,\phi)$. In order to have an atlas of $M_f$ at our disposal, we use a second coordinate system $(T,x,y,\psi)$ with $\psi$ mapping into $(0,2\pi)$. We shall refer to the first coordinate system as A and to the second as B. The transformation formula
\begin{equation}
 \psi(\phi)	 = \left\{
 \begin{aligned} 
  & \phi   	& \qquad	 	& \text{if}  &  0 < & \, \phi<\pi, 		\\
  & \phi + 2\pi  &		& \text{if}  &  -\pi < & \, \phi < 0,
 \end{aligned} \right.	
\end{equation}
implies the coordinate expression
\begin{equation}
K	=	\left\{
\begin{aligned}
 & e^{-\frac{\psi}{2}} \partial_T 	\qquad	& 	& \text{if} &	0< & \, \phi<\pi,	 \\
 & e^{-\frac{\psi}{2}-\pi} \partial_T 		& 	& \text{if} &  -\pi< & \, \phi<0,	
\end{aligned} \right.
\end{equation}
for $K$ in system $B$. If $K$ were extendible to all of $M_f$, using B, the two limits
\begin{align}
 \lim_{\psi \nearrow \pi} K_p &= e^{-\frac{\pi}{2}} \partial_T	\\
\intertext{and}
 \lim_{\psi \searrow \pi} K_p &= \lim_{\psi \searrow \pi} e^{-\frac{\psi}{2}-\pi} \partial_T = e^{-\frac{3\pi}{2}} \partial_T
\end{align}
would necessarily be the same. 
Next we verify that $K$ is indeed a Killing vector field. Generally, for a vector field $X \pdef G\partial_T + H\partial_x + P\partial_y + Q\partial_\phi$ the Killing equation $\mathcal{L}_Xg=0$ is equivalent to the following system of PDE's:
\begin{gather}
\begin{aligned}
 Q_T						& = 0		\\
 H_T - Q_x					& = 0		\\
 P_T - Q_y					& = 0		\\
 P_y						& = 0		\\
 H_x						& = 0		\\
 P_x + H_y					& = 0		
\end{aligned} \notag \\
\begin{aligned}
 G_T - f Q_T + T Q_T  + Q_\phi	  	& = 0		\\
 G_x - f Q_x + T Q_x  - H_\phi	  	& = 0		\\
 G_y - f Q_y + T Q_y  - P_y		& = 0		\\
 -2G_\phi + 2Q_\phi f + Q f_\phi -2 T Q_\phi + H f_x & + P f_y - G = 0
\end{aligned} \notag
\end{gather}
The assumption $H = P = Q = 0$ and the supposition that $G$ depends only on $\phi$ reduces this system of PDEs to the simple ordinary differential equation 
\begin{equation}
 G + 2 G_\phi =0.
\end{equation}
Hence $G(\phi) = e^{-\frac{\phi}{2}}$ and $K$ is a Killing vector field. 
It remains to show that there is a choice for $f$, such that there is no other solution to the Killing equation which is not a constant multiple of $K$. This is the point where Prop.~\ref{prop-dimension-Killing} comes in. If we choose 
\begin{equation}
 f(x,y,\phi) \pdef e^{xy} \cos \phi,
\end{equation}
the equation $\mathcal L_X DR = 0$ evaluated at $(T,x,y,\phi)=(1,1,1,0)$ becomes a linear system with kernel of dimension $1$.
\end{proof}
The fact that $K$ cannot be extended to a global Killing field is due to the topology of $M_f$. If we consider the universal covering $k \colon \R^4 \to M_f$ of $M_f$, assign to $\R^4$ the pull-back metric $k^*g$ and use global coordinates $(T,x,y,z)$ on $\R^4$, the vector field
\begin{equation}
 \tilde K \pdef e^{-\frac{z}{2}} \partial_z
\end{equation}
is a global Killing field on $\R^4$, irrespective of the explicit form of $f$. Changing the $z$-coordinate to the periodic coordinate $\phi$ implies loosing the global Killing field $\tilde K$.

\subsection{One Special Choice for $f$}
We are now going to consider a special example of the class of spacetimes described so far by choosing $f$ to be
\begin{equation}
 f(x,y,\phi) \pdef \frac{a}{2}(x^2-y^2),
\end{equation}
where $a$ is a positive constant. For notational convenience we drop the subscript $f$ in $M_f$ from now on. 
In order to obtain a coordinate system which is better suited to the description of the region of CTCs in $M$ we use the transformation:
\begin{equation}
 \R^3 \times \mathbb{S}^1 \to \R^3 \times \mathbb{S}^1, \qquad (T,x,y,\phi) \mapsto (t,x,y,\phi) 
\end{equation}
with 
\begin{equation}
t \pdef T - \frac{a}{2}(x^2-y^2) + e \rho^2. 
\end{equation}
Here, $e$ is another positive constant and $\rho^2 \pdef x^2 + y^2$. It is clear that $(t,x,y,\phi)$ is a coordinate system. A short calculation yields 
\begin{multline} \label{equ-metric-2005-new}
 g = - dt \otimes_s d\phi + (2e-a)x dx \otimes_s d\phi + (2e+a)y dy \otimes_s d\phi \\
     + dx \otimes dx + dy \otimes dy + (e \rho^2 -t) d\phi \otimes d\phi
\end{multline}
in the new coordinates.

\subsubsection{CTCs in $M$}
For fixed values of $t$, $x$ and $y$ the curves $\gamma_{t,x,y}(s)= (t,x,y,s)$ are closed and 
\begin{equation}
 \begin{cases}
  \text{spacelike} \qquad  &\text{for} \quad t<e\rho^2,		\\
  \text{null}		   &\text{for} \quad t=e\rho^2,		\\
  \text{timelike}	   &\text{for} \quad t>e\rho^2.
 \end{cases}
\end{equation}
For $t<0$ all these curves are spacelike. At $t=0$ the curve $\gamma_{0,0,0}$ is a closed null geodesic, as will follow from later results. 
For a given value of $t$ all curves $\gamma_{t,x,y}$ lie in the hypersurface $H_t$ given by $t=const$. The causal character of $H_t$ can be inferred from
\begin{equation}
 g^{tt} = t + (2e-a)^2 x^2 + (2e+a)^2 y^2 - e\rho^2.
\end{equation}
If we suppose
\begin{equation}
 (2e+a)^2 < e,
\end{equation}
then $H_t$ is spacelike throughout for $t<0$. At $t=0$, $H_t$ is spacelike except at the closed null geodesic $x=y=0$. In the region $t>0$,
\begin{equation}
 H_t \quad \text{is} \quad
\begin{cases}
 \text{spacelike} \qquad   &\text{if $(x,y)$ lies outside $E_t$,}		\\
 \text{timelike}	   &\text{if $(x,y)$ lies inside $E_t$,}
\end{cases}
\end{equation}
where $E_t$ is the ellipse in the $x$-$y$ plane described by $g^{tt}=0$ or
\begin{equation}
 t = [e - (2e-a)^2] x^2 + [e- (2e+a)^2] y^2.
\end{equation}
Thus any CTC in $H_t$ must be entirely contained in $E_t$ (see Fig. \ref{fig-CTCs-2005}).

\begin{figure}[htb]
 \includegraphics{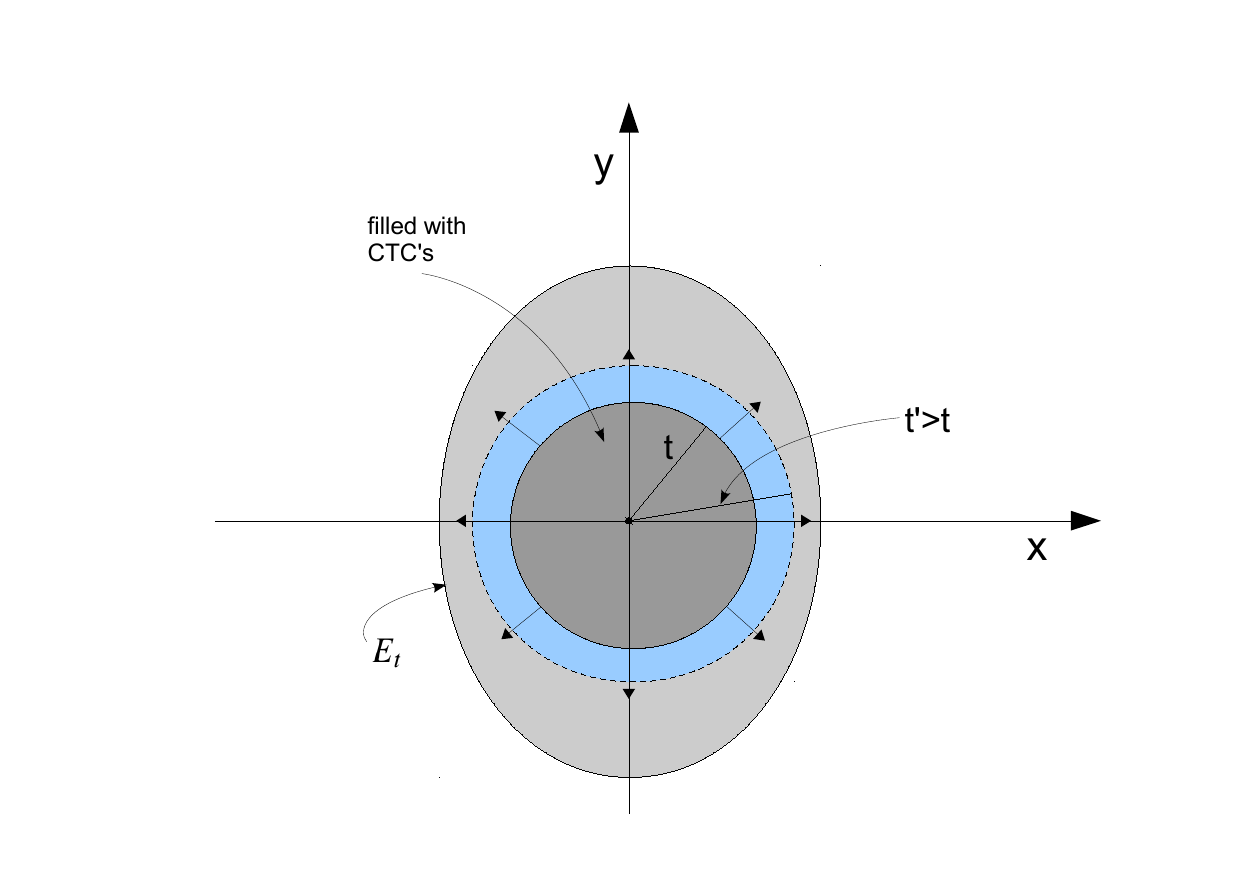}
 \caption{A section $\phi = const$ in $H_t$ and the circle of radius $t$ in the $x$-$y$ plane filled with CTCs of the form $\gamma_{t,x,y}$ (we have set $e=1$). The growth of this circle with $t$ is indicated by the arrows. The ellipse $E_t$ confines the region of CTCs in $H_t$.}
 \label{fig-CTCs-2005}
\end{figure}
In \cite{Ori05}, Ori describes in some detail how to construct a spacelike hypersurface $\mathscr S$ in the region $\{p \in M \colon t(p) \leq 0 \}$ of $M$ that contains a compact subset $\mathscr C$, such that the closed null curve $\gamma_{0,0,0}$ is contained in $\overline{D^+(\mathscr C)}$, making $M$ a candidate for a spacetime with TM-structure.

\subsubsection{Killing Vector Fields of $M$}
We have already seen that due to the topology of $M$ the Killing field $K$, which exists on $M$ independently of the explicit form of $f$, is not a global Killing field. We will be accompanied by this problem throughout this section and, therefore, we concentrate our efforts on the part $U$ of $M$ covered by the coordinates $(t,x,y,\phi)$, where $\phi \in (-\pi,\pi)$. However, since the expression for $f$ does not contain the coordinate $\phi$ we immediately conclude that the vector field $\partial_\phi$ is a global Killing field on $M$; it will be the only one.

\begin{proposition} \label{prop-Killing-2005-new}
 The dimension of $i(U)$ satisfies
 \begin{equation}
  \dim i(U) \leq 6
 \end{equation}
 and $K_6 = \partial_\phi$ and $K_5 = e^{-\frac{\phi}{2}} \partial_t$ are Killing fields on $U$.
\end{proposition}

\begin{proof}
 We use Prop.~\ref{prop-dimension-Killing} to get the upper bound for $\dim i(U)$. At the point $p=(1,1,1,0)$ the linear system $\mathcal L_X R =0$ consists of the following four linearly independent equations:
\begin{align*}
 (L_X R)_{4241} &= 		&	& \frac{a}{2} X_{1;2}	&	&			&	& & &	\\
 (L_X R)_{4341} &= 		&	&			&	& -\frac{a}{2} X_{1;3}	&	& & &		\\
 (L_X R)_{4242} &= -\frac{a}{2} X^4 &	& -(2e-a)a X_{1;2} 	&	&			&	& + a X_{1;4} & &	\\
 (L_X R)_{4342} &=		&	& \frac{a}{2} X_{1;2}   &	& -\frac{a}{2}(2e-a) X_{1;3} & 	& & & - a X_{2;3}
\end{align*}
Hence, the result follows. 
Since $\partial_t = \partial_T$, the vector field $K_5$ is a Killing field because of Prop. \ref{killgen}.
\end{proof}
The indices of $K_5$ and $K_6$ in the foregoing proposition indicate that indeed $\dim i(U)=6$. In order to find the other four Killing fields we have to write down explicitly the equations given by the Killing equation $\mathcal{L}_Xg=0$. If we set
\begin{equation}
 X \pdef G \partial_t + H \partial_x + P \partial_y + Q \partial_\phi,
\end{equation}
these are:
\begin{align}
 Q_t 								&= 0		\label{k1} \\
 H_t + (2e-a) x Q_t - Q_x 					&= 0		\label{k2} \\
 P_t + (2e+a) y Q_t - Q_y 					&= 0		\label{k3} \\
 -G_t + (2e-a) x H_t + (2e+a) y P_t + (e\rho^2 -t) Q_t - Q_\phi &= 0		\label{k4} \\
 H_x + (2e-a) x Q_x 						&= 0		\label{k5} \\
 P_x + (2e+a) y Q_x + H_y + (2e-a) x Q_x 			&= 0		\label{k6} \\
 -G_x + (2e-a)H + (2e-a)x H_x + (2e+a)y P_x 			& 		\notag	\\
 +(e\rho^2 -t) Q_x + H_\phi + (2e-a)x Q_\phi 			&= 0		\label{k7} \\
 P_y + (2e+a)y Q_y						&= 0		\label{k8} \\
 -G_y + (2e-a)x H_y + (2e+a)P + (2e+a)y P_y			&		\notag	\\
 +(e\rho^2 -t) Q_y + P_\phi + (2e+a)y Q_\phi 			&= 0	 	\label{k9} \\
 -G -2 t Q_\phi + 2e(xH + yP) + 2e\rho^2 Q_\phi 		&		\notag	\\
 +2(2e+a)y P_\phi + 2(2e-a)x H_\phi - 2G_\phi 			&= 0		\label{k10}
\end{align}
The following two assumptions will simplify this system:

\begin{itemize}
 \item $Q = 0$;
 \item $H$ and $P$ are functions of $\phi$ only.
\end{itemize}
Then eqs. (\ref{k1})-(\ref{k3}), \eqref{k5}, \eqref{k6} and eq. (\ref{k8}) are trivially satisfied and we are left with:
\begin{gather}
  G_t = 0					\label{k4'}	\\
 -G_x + (2e-a)H + H_\phi = 0	\label{k7'} 			\\
 -G_y + (2e+a)P + P_\phi = 0	\label{k9'} 			\\
 -G + 2e(xH + yP) + 2(2e+a)y P_\phi + 2(2e-a)x H_\phi -2G_\phi = 0 \label{k10'}
\end{gather}
At this point there are two similar cases.\\

\emph{1.} Suppose $P = 0$. The last three equations simplify to:
\begin{gather}
 G_y = 0	\label{k9''} \\
 -G_x + (2e-a)H + H_\phi = 0	\label{k7''} \\
 -G + 2exH + 2(2e-a)xH_\phi - 2G_\phi = 0	\label{k10''}
\end{gather}
Because of (\ref{k4'}) and (\ref{k9''}), $G$ does not depend on $t$ or $y$. Since $H$ is a function of $\phi$ only, we deduce from (\ref{k7''}) that 
\begin{equation}
 G(x,\phi) = x (H_\phi(\phi) + (2e-a)H(\phi)),
\end{equation}
where we have set a possible integration constant equal to zero. Substituting this into (\ref{k10''}) results in the ODE
\begin{equation} \label{DEh}
 2H_{\phi \phi} + H_\phi - aH = 0.
\end{equation}
The ansatz $H(\phi) = e^{-\lambda \phi}$ leads to the polynomial equation
\begin{equation}
 2 \lambda^2 - \lambda - a = 0
\end{equation}
with zeros
\begin{equation} \label{albe}
 \alpha \pdef \frac{1}{4} \left( 1 + \sqrt{1 + 8a} \right) \quad \text{and} \quad \beta \pdef \frac{1}{4} \left( 1 - \sqrt{1 + 8a} \right),
\end{equation}
and the solutions to \eqref{DEh} are
\begin{equation}
 H_1(\phi) \pdef e^{-\alpha \phi} \quad \text{and} \quad H_2(\phi) \pdef e^{-\beta \phi}.
\end{equation}
For $G$ we calculate
\begin{equation}
 G_1(x,\phi) \pdef (2e -a - \alpha)x e^{-\alpha \phi} \quad \text{and} \quad G_2(x,\phi) \pdef (2e -a - \beta)x e^{-\beta \phi}.
\end{equation}
With these choices for $H$ and $G$ all 10 components of the Killing equation are satisfied. \\

\emph{2.} Suppose $H = 0$. From (\ref{k7'}), (\ref{k9'}) and (\ref{k10'}) we now get:
\begin{gather}
 G_x = 0	\\
 -G_y + (2e+a)P + P_\phi = 0	\\
 -G + 2eyP + 2(2e+a)yP_\phi - 2G_\phi = 0
\end{gather}
Reasoning as in the first case leads to
\begin{equation}
 G(y,\phi) = y(P_\phi(\phi) + (2e+a)P(\phi))
\end{equation}
and an ODE for $P$:
\begin{equation}
 2P_{\phi \phi} + P_\phi + aP = 0
\end{equation}
The solutions for this equation depend on the value of $a$.\\

\emph{A.} For $0<a<\frac{1}{8}$ they take the following form:
\begin{equation}
 P_{11}(\phi) \pdef e^{-\mu \phi} \quad \text{and} \quad P_{12}(\phi) \pdef e^{-\nu \phi},
\end{equation}
where
\begin{equation} \label{munu}
 \mu \pdef \frac{1}{4} \left(1 + \sqrt{1-8a} \right) \quad \text{and} \quad \nu \pdef  \frac{1}{4} \left(1 - \sqrt{1-8a} \right).
\end{equation}
The corresponding expressions for $G$ are
\begin{equation}
 G_{11}(y,\phi) \pdef (2e+a- \mu)y e^{-\mu \phi} \quad \text{and} \quad G_{12}(y,\phi) \pdef (2e+a - \nu)y e^{-\nu \phi}.
\end{equation}
\emph{B.} If $a=\frac{1}{8}$, the solutions are
\begin{align}
 P_{21}(\phi) &\pdef e^{-\frac{\phi}{4}}; & P_{22}(\phi) &\pdef \phi e^{-\frac{\phi}{4}};	\\
 G_{21}(y,\phi) &\pdef (2e - \frac{1}{8})y e^{-\frac{\phi}{4}}; & G_{22}(y,\phi) &\pdef ((2e-\frac{1}{8}) \phi +1)y e^{-\frac{\phi}{4}};
\end{align}
and finally:\\

\emph{C.} For $a>\frac{1}{8}$ we have
\begin{equation}
 P_{31}(\phi) \pdef e^{-\frac{\phi}{4}} \sin \Delta(\phi), \qquad P_{32}(\phi) \pdef e^{-\frac{\phi}{4}} \cos \Delta(\phi)	
\end{equation}
and
\begin{align}
 G_{31}(y,\phi) &\pdef \left( \frac{1}{4} \sqrt{8a-1} \cos \Delta(\phi) + (2e + a - \frac{1}{4}) \sin \Delta(\phi) \right)y e^{-\frac{\phi}{4}}, \\ 
 G_{32}(y,\phi) &\pdef \left( -\frac{1}{4} \sqrt{8a-1} \sin \Delta(\phi) + (2e + a - \frac{1}{4}) \cos \Delta(\phi) \right)y e^{-\frac{\phi}{4}}
\end{align}
with $\Delta(\phi) \pdef \frac{1}{4} \sqrt{8a-1} \, \phi$.\\

We summarize:

\begin{proposition}
 A basis for $i(U)$ is given by the following Killing fields:
 \begin{align}
  K_1 &= e^{-\alpha \phi} [ (2e-a-\alpha)x \partial_t + \partial_x ]	\notag \\
  K_2 &= e^{-\beta \phi} [ (2e-a-\beta)x \partial_t + \partial_x ]	\notag \\
  K_3 &= G_{i1} \partial_t + P_{i1} \partial_y				\notag \\
  K_4 &= G_{i2} \partial_t + P_{i2} \partial_y				\notag \\	
  K_5 &= e^{-\frac{\phi}{2}} \partial_t					\notag \\
  K_6 &= \partial_\phi							\notag \\
\intertext{Here, the $G_{ik}$ and $P_{ik}$ are the functions defined above and}
i &= 
 \begin{cases}
  1 & \quad \text{if} \quad 0<a<\frac{1}{8},	\\
  2 & \quad \text{if} \quad a=\frac{1}{8},		\\
  3 & \quad \text{if} \quad a>\frac{1}{8}.
 \end{cases}\notag 
\end{align}
\end{proposition}
Note that, as in the case of the Killing field $K$ in Prop. \ref{killgen}, none of the Killing fields $K_1$ - $K_5$ can be extended to all of $M$.

\subsubsection{Geodesics of $M$}
As we shall see, with the given special form of $f$, it is possible to solve the geodesic equations on $M$ analytically. We express the geodesic $\gamma \colon I \to M$ defined on some interval $I$ around zero as $\gamma(s) = (t(s),x(s),y(s),\phi(s))$, where $s$ is a proper time parameter in the case of a timelike geodesic $\gamma$ and an affine parameter if $\gamma$ is lightlike or spacelike. Then the geodesic equations read:
\begin{gather}
 \ddot \phi - \frac{1}{2} \dot \phi ^2		=0		\label{gep}					\\
 \ddot x -\frac{1}{2} a x \dot \phi ^2		=0		\label{gex}					\\
 \ddot y + \frac{1}{2} a y \dot \phi ^2	=0		\label{gey}					\\	
 \ddot t + \dot t \dot \phi -(2e-a) \dot x^2 -(2e+a) \dot y^2 -2e(x \dot x + y \dot y) \dot \phi 	\notag  \\
  + \frac{1}{2} \left( t + (a^2 -2ea -e) x^2 + (a^2 + 2ea -e) y^2 \right) \dot \phi^2  =0		\label{get}
\end{gather}
Furthermore, we have the condition
\begin{equation} \label{mc2005}
 -2 \dot t \dot \phi + \dot x^2 + \dot y ^2 +2(2e-a)x \dot x \dot \phi +2(2e+a)y \dot y \dot \phi + (e(x^2+y^2) - t) \dot \phi ^2 = k,
\end{equation}
where $k=-1,0,1$ for timelike, lightlike or spacelike geodesics, respectively.


The solving process consists of two major steps. In the first step we solve eqs. (\ref{gep})-(\ref{gey}), which are independent of $t$, and in the second step we use \eqref{mc2005} to determine an expression for $t(s)$.\\

\emph{Step 1.} Use the Killing field $K_6$ or directly integrate \eqref{gep} to get
\begin{equation} \label{gep'}
 e^{-\frac{\phi}{2}} \dot \phi = A
\end{equation}
for an arbitrary constant $A$. Note that \eqref{gep'} is equivalent to \eqref{gep}. The general solution for \eqref{gep'} in the case of $A \neq 0$ is
\begin{equation} \label{solp}
 \phi(s) = - 2 \ln \left( - \frac{A}{2} (s + B) \right),
\end{equation}
where $B$ is another constant. The case $A =0$ will be treated seperately later on.
Since we would like $\gamma$ to be defined on some interval around zero, we assume $B \neq 0$. Because $A$ has the same sign as $\dot \phi(0)$, we see that $\phi$ is only defined on
\begin{enumerate}
 \item $(-\infty,-B)$ if $\dot \phi (0) >0$, or
 \item $(-B,\infty)$ if $\dot \phi (0) <0$.
\end{enumerate}
In either case $\gamma$ is incomplete. By traversing $\gamma$ in the opposite direction, i.e., by setting $\tilde \gamma (s) \pdef \gamma(-s)$, it is always possible to obtain the case $\dot \phi (0) <0$. Then $s+B>0$ on $I$ and since this will simplify the calculations to come, we make the assumption
\begin{equation}
 \dot \phi (0) <0 \qquad \Leftrightarrow \qquad B>0 \qquad \Leftrightarrow \qquad A<0.
\end{equation}
If we substitute the result \eqref{solp} into \eqref{gex}, we arrive at
\begin{equation} \label{gex'}
 \ddot x - \frac{2a}{(s+B)^2} x = 0.
\end{equation}
The solution of this equation is
\begin{equation} \label{solx}
 x(s) = C (s+B)^{2 \alpha} + D (s+B)^{2\beta}
\end{equation}
with $\alpha$ and $\beta$ defined in \eqref{albe} and two further constants $C$ and $D$. 
Similarly, \eqref{gey} becomes
\begin{equation} \label{gey'}
 \ddot y + \frac{2a}{(s+B)^2} y = 0,
\end{equation}
but the solution to this equation depends on the value of $a$. With constants $E$, $F$ and the function
\begin{equation}
 \kappa(s) \pdef \frac{1}{2} \sqrt{8a-1} \ln(s+B)
\end{equation}
we have (for the constants $\mu$ and $\nu$ confer \eqref{munu}):
\begin{align}
 0<a<\frac{1}{8} & \colon & \qquad y(s) = & E (s+B)^{2\mu} + F (s+B) ^{2\nu}	\label{soly1} \\
 a = \frac{1}{8} & \colon & 	y(s) = & \sqrt{s+B} (E + F \ln (s+B))		\label{soly2} \\
 a>\frac{1}{8}   & \colon &	y(s) = & \sqrt{s+B} \left[ E \sin(\kappa(s)) + F \cos (\kappa(s)) \right] \label{soly3}
\end{align}
\emph{Step 2.} To find an expression for $t(s)$ is more complicated. First, we rewrite the metric condition \eqref{mc2005} as
\begin{equation} \label{mc2005'}
 k + 2 \dot t \dot \phi + t \dot \phi ^2 = h,
\end{equation}
where
\begin{equation} \label{funch}
 h \pdef \dot x^2 + \dot y ^2 +2(2e-a)x \dot x \dot \phi +2(2e+a)y \dot y \dot \phi +(e(x^2+y^2) - t) \dot \phi ^2.
\end{equation}
This equation is easily integrated. One obtains
\begin{equation} \label{soltimp}
 t(s) = \left[ \int \frac{1}{4} \left(k - h(s) \right) ds + G \right] (s+B),
\end{equation}
$G$ being another constant.
For each of the three possibilities for $a$ we calculate $h(s)$ and perform the integration in \eqref{soltimp}. Since the explicit steps of the calculation are rather long but straightforward we suppress them and simply state the results.\\

\emph{A.} In the case of $0<a<\frac{1}{8}$ we have
\begin{equation} \label{solt1}
 t(s) = \bar t(s) + u_3 (s+B)^{4 \mu} + u_4 (s+B)^{4 \nu} 
\end{equation}
with 
\begin{equation}
 \bar t(s) \pdef \frac{k}{4} s (s+B) + G(s+B) + u_1 (s+B)^{4\alpha} + u_2 (s+B)^{4\beta}
\end{equation}
and the abbreviations
\begin{equation}\begin{aligned}
 u_1 &\pdef \frac{C^2}{2}(2e-a-\alpha),	&	u_2 &\pdef \frac{D^2}{2}(2e-a-\beta),	\\
 u_3 &\pdef \frac{E^2}{2}(2e+a-\mu),		&	u_4 &\pdef \frac{F^2}{2}(2e+a-\nu).
\end{aligned} \end{equation}
\emph{B.} Here $a= \frac{1}{8}$ and
\begin{equation} \label{solt2}
 t(s) = \bar t(s) + (v_1 + v_2 \ln(s+B)) (s+B) \ln(s+B),
\end{equation}
where
\begin{equation}
 v_1 \pdef \frac{1}{4} F (E(8e - \frac{1}{2}) - F ) \quad \text{and} \quad v_2 \pdef \frac{1}{8} F^2 ( 8e - \frac{1}{2}).
\end{equation}
Finally, there is:\\

\emph{C.} For $a> \frac{1}{8}$ the solution takes the form
\begin{equation} \label{solt3}
 t(s) = \bar t(s) + \left[w_2 \cos^2 (\kappa(s)) - w_1 \cos \kappa(s) \sin \kappa(s) \right] (s+B),
\end{equation}
where the constants $w_1$ and $w_2$ are
\begin{equation}\begin{aligned}
 w_1 &\pdef \frac{1}{4} \left[ \frac{1}{2} \sqrt{8a-1}(E^2-F^2) - (4(2e+a) -1)EF \right],		\\
 w_2 &\pdef \frac{1}{4} \left[ \left( \frac{1}{2} - 2(2e+a) \right) (E^2-F^2) - \sqrt{8a-1} EF \right].
\end{aligned} \end{equation}
In order to complete our programme we need
\begin{proposition}
 Any curve $\gamma$ satisfying the first three geodesic equations \eqref{gep} - \eqref{gey} and the metric condition \eqref{mc2005} is a geodesic.
\end{proposition}
\begin{proof}
 We have to check that the geodesic equation for $t$ is satisfied by $\gamma$. If we set
\begin{equation}
 \gamma'' = a^{(t)} \partial_t + a^{(x)} \partial_x + a^{(y)} \partial_y + a^{(\phi)} \partial_\phi,
\end{equation}
then equations \eqref{gep}-\eqref{gey} give $a^{(x)} = a^{(y)} = a^{(\phi)} = 0$. Derivation of the metric condition $\langle \gamma', \gamma' \rangle = k$ w.r.t. $s$ results in
\begin{equation}
 \langle \gamma'', \gamma' \rangle = 0 \qquad \Leftrightarrow \qquad \dot \phi \, a^{(t)} = 0.
\end{equation}
As we know from \eqref{gep'}, the function $\dot \phi$ is either always zero or never zero. Thus, if $\dot \phi$ is never zero, then $a^{(t)} = 0$ and $\gamma$ is a geodesic. In a moment we will see that the case $\dot \phi = 0$ will lead to a geodesic as well.
\end{proof}

\emph{The case $A=0$.} This condition is equivalent to $\dot \phi = 0$ and therefore we immediately conclude
\begin{equation}
 x(s) = c_1 s + c_0 \qquad \text{and} \qquad y(s) = d_1 s + d_0.
\end{equation}
The metric condition now reads
\begin{equation}
 c_1^2 + d_1^2 =k.
\end{equation}
Hence for a causal geodesic, we have $k=c_1=d_1=0$. Then the geodesic equation for $t$ yields $\ddot t =0$, thus $\gamma$ is of the form
\begin{equation}
 \gamma(s) = (l_1 s + l_0,x_0,y_0,\phi_0),
\end{equation}
for constants $l_0$ and $l_1$, i.e., linear parametrizations of $t$-coordinate lines are null geodesics.
If $k=1$, we find that 
\begin{equation}
 \gamma(s) = (l_2 s^2 + l_1 s + l_0 , c_1 s + c_0, d_1 s + d_0, \phi_0),
\end{equation}
where
\begin{equation}
 l_2 \pdef (2e-a) c_1^2 + (2e+a) d_1 ^2 \qquad \text{and} \qquad c_1^2 + d_1^2 =1,
\end{equation}
are spacelike geodesics. 

\subsubsection{Discussion}
To simplify later expressions, we introduce the notation
\begin{equation}
 \isgn(A) \pdef \begin{cases}
                 + \infty, \qquad 	& \text{if } A> 0,	\\
		 0			& \text{if } A=0,	\\
		 - \infty		& \text{if } A <0,
                 \end{cases}
\end{equation}
for any real number $A$. 
As we see from the solutions, the constant $B$ determines the domain of $\gamma$. Since we assume $\dot \phi (0) <0$, we always have $\gamma \colon (-B, \infty) \to M$. The causal character of $\gamma$ does not have any influence on the behavior of the three component functions $\phi$, $x$ and $y$. Hence the following holds for any geodesic.\\

\emph{Behavior of $\phi$.}\\
From the result \eqref{solp} obtained above we determine
\begin{equation}
 \dot \phi(s) = - \frac{2}{s+B}, \qquad B = - \frac{2}{\dot \phi(0)}, \qquad A = \dot \phi(0) e^{-\frac{\phi(0)}{2}}.
\end{equation}
Hence on $I=(-B, \infty)$ the function $\phi$ is strictly decreasing and its asymptotic behavior is (cf. Fig. \ref{fig-phi-x})
\begin{equation}
 \lim_{s \searrow -B} \phi(s) = + \infty \qquad \text{and} \qquad \lim_{s \to \infty} \phi(s) = - \infty.
\end{equation}
However, we must not forget that $\phi$ is a circular coordinate, i.e., geometrically important is the behavior of $s \mapsto \varrho(\phi(s)) \pdef (\cos \phi(s), \sin \phi(s) ) \in \mathbb{S}^1$. In particular, this means that as $s$ approaches $-B$ from above, $\gamma$ circles infinitely often around $\mathbb{S}^1$ in $M=\R^3 \times \mathbb{S}^1$.\\

\emph{Behavior of $x$.}\\
Since $\alpha >0$ and $\beta <0$, we conclude from \eqref{solx} that
\begin{equation}
 \lim_{s \searrow -B} x(s) = \isgn (D)  \qquad \text{and} \qquad \lim_{s \to \infty} x(s) = \isgn(C).
\end{equation}
The initial values $x(0)$ and $\dot x(0)$ determine the parameters $C$ and $D$ by
\begin{equation} \label{invaluesx}
\begin{aligned}
 C &= \frac{1}{\sqrt{1+8a} B^{2\alpha - 1}} \left( \dot x(0) - \frac{2 \beta}{B} x(0) \right), \\
 D &= -\frac{1}{\sqrt{1+8a} B^{2\beta - 1}} \left( \dot x(0) - \frac{2 \alpha}{B} x(0) \right). 
\end{aligned}
\end{equation}

\begin{figure}[here] \centering
 \includegraphics[scale=0.5]{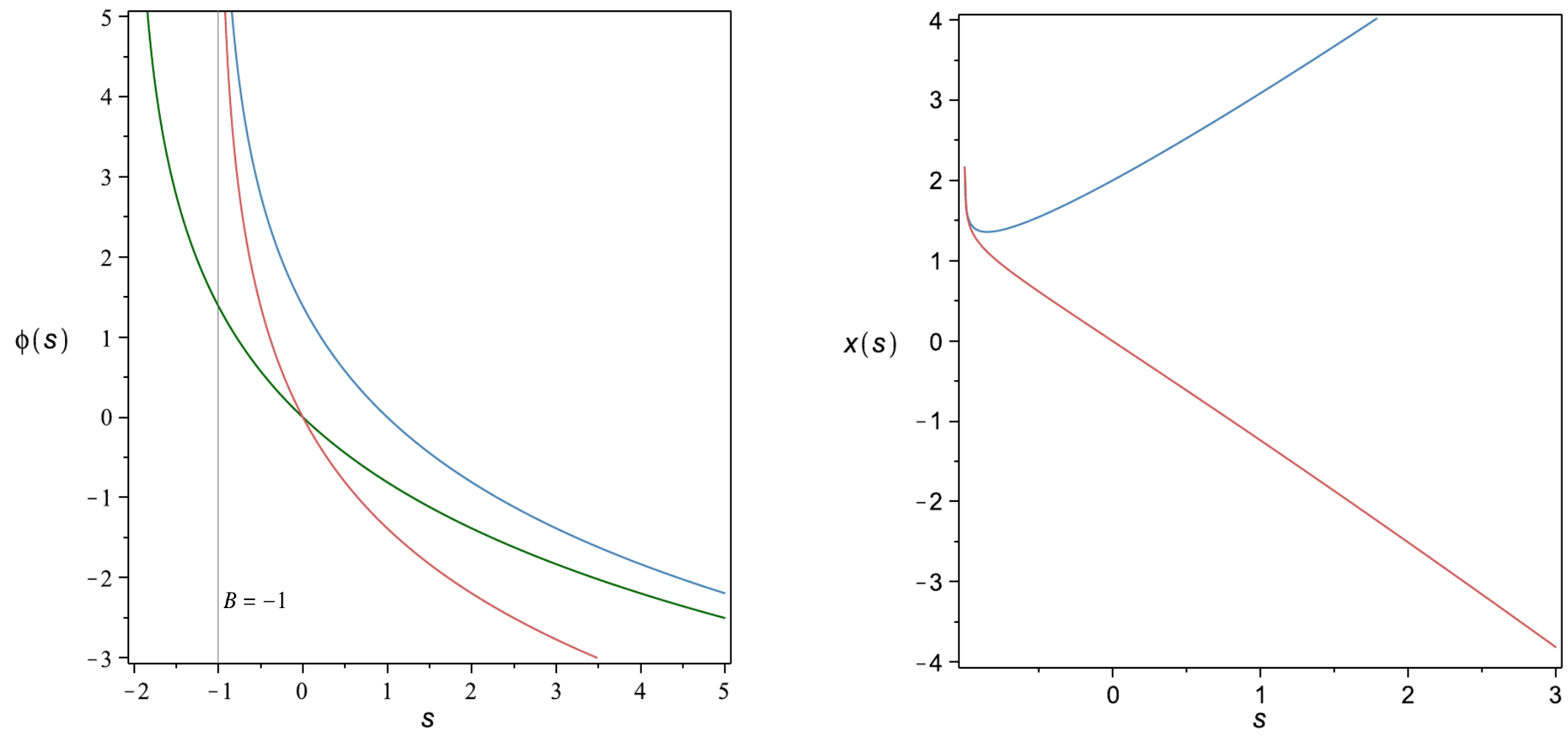}
 \caption{Some examples of the geodesic component functions $\phi$ and $x$. On the left the parameter values are $A=-1$, $B=1$ (blue), $A=-2$, $B=1$ (red) and $A=-1$, $B=2$ (green). For $B=-1$ the asymptote is shown. With $a=\frac{1}{16}$ (right), the function $x$ is displayed for $C=D=1$ (blue) and $C=-1$, $D=1$ (red).}
 \label{fig-phi-x}
\end{figure}

\emph{Behavior of $y$.}\\
Here again, we have to distinguish the already well known cases for $a$.\\
\emph{A.} Because of $0 < 4 \nu < 1< 4 \mu$, \eqref{soly1} implies
\begin{equation}
 \lim_{s \searrow -B} y(s) = 0 \qquad \text{and} \qquad \lim_{s \to \infty} y(s) = 
	\begin{cases}
	 \isgn(E) \qquad & \text{if} \quad E \neq 0, \\
	 \isgn(F) 	& \text{if} \quad E=0.
	\end{cases}
\end{equation}
The parameters $E$ and $F$ are related to the initial values $y(0)$ and $\dot y(0)$ by formulas completely analogous to \eqref{invaluesx}.\\

\emph{B.} It follows from \eqref{soly2} that
\begin{equation}
 \lim_{s \searrow -B} y(s) = 0 \qquad \text{and} \qquad \lim_{s \to \infty} y(s) = 
	\begin{cases}
	 \isgn(F) \qquad & \text{if} \quad F \neq 0, \\
	 \isgn(E) 	& \text{if} \quad F=0.
	\end{cases}
\end{equation}
The relevant parameter $F$ can be obtained by 
\begin{equation}
 F = \sqrt{B} \left( \dot y(0) - \frac{y(0)}{2B} \right).
\end{equation}

\begin{figure}[here] \centering
 \includegraphics[scale=0.5]{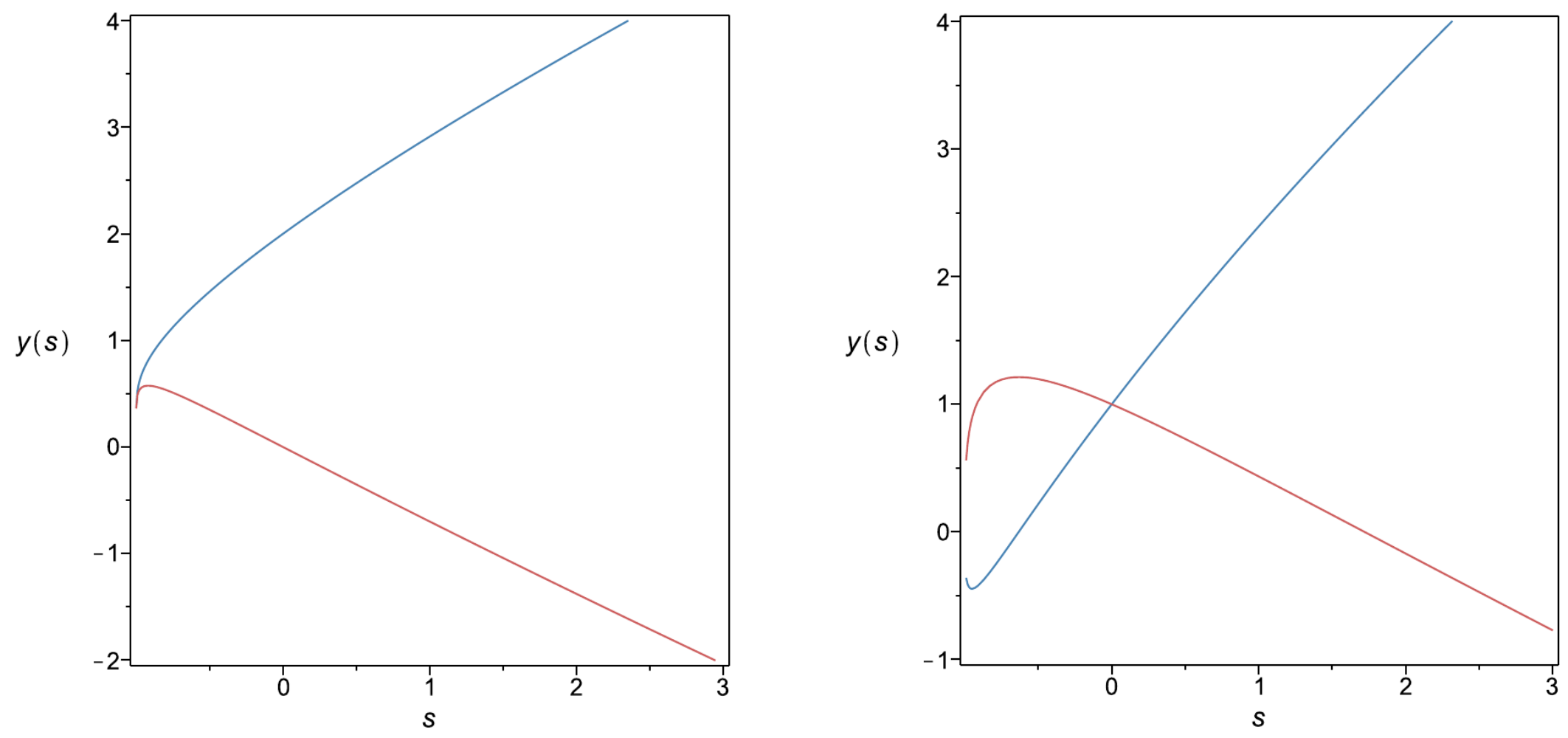}
 \caption{The geodesic component function $y$. For $a=\frac{1}{16}$ (left) we chose $E=F=1$ (blue) and $E=-1$, $F=1$ (red). On the right $a=\frac{1}{8}$ and the parameters are $E=F=1$ (blue) as well as $E=1$, $F=-1$ (red).}
\end{figure}

\emph{C.} In this case we refer to \eqref{soly3}. Still
\begin{equation}
 \lim_{s \searrow -B} y(s) = 0,
\end{equation}
but for $s \to \infty$, the function $y(s)$ does not converge, in general. Instead it oscillates consecutively between the values $F$, $E$, $-F$ and $-E$ multiplied with $\sqrt{s+B}$.\\

\begin{figure}[here] \centering
 \includegraphics[scale=0.3]{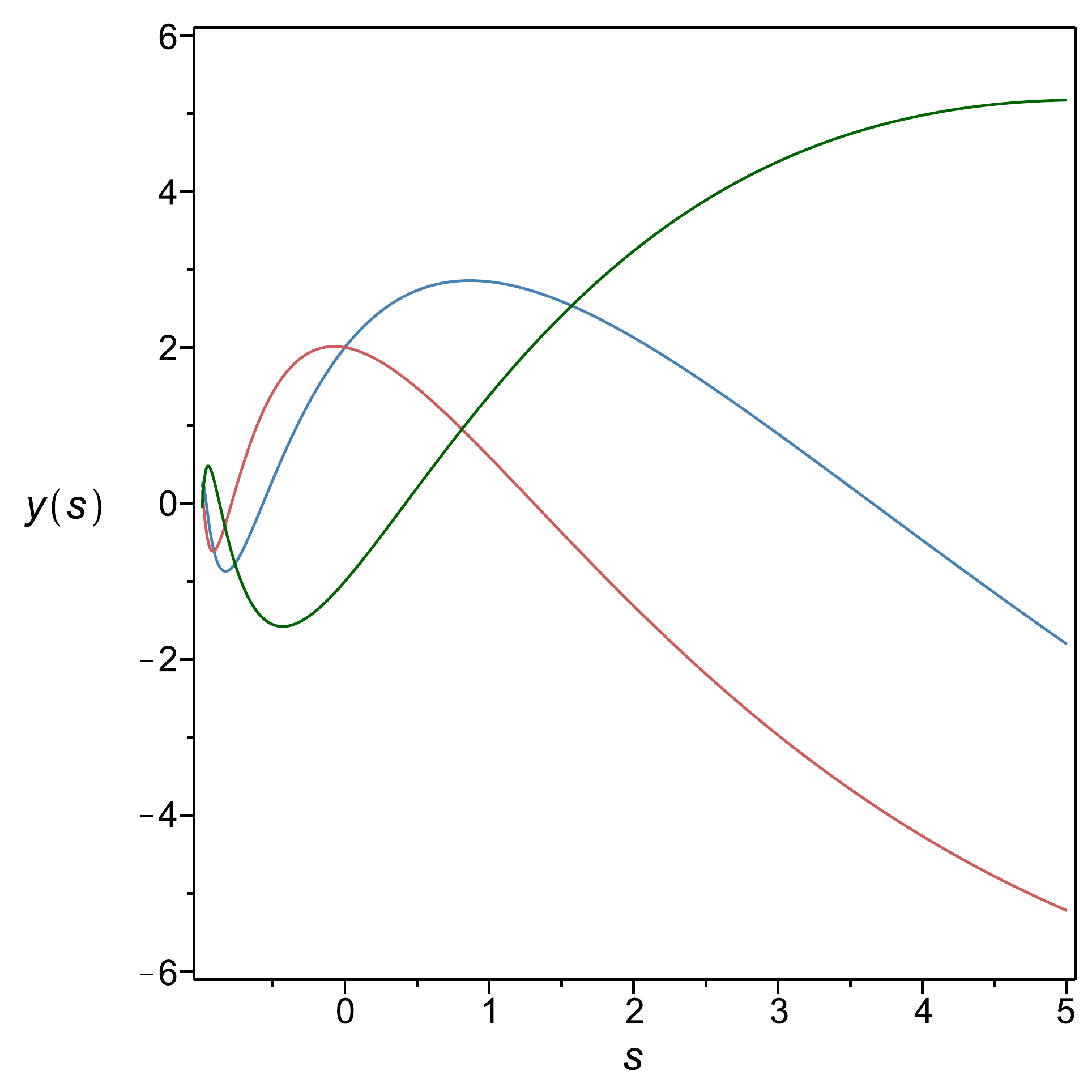}
 \caption{When $a=1$ we get oscillating behavior for $y(s)$. $E=1$, $F=2$ is in blue, $E=-1$, $F=2$ in red and $E=2$, $F=-1$ green.}
\end{figure}

\emph{Behavior of $t$.}\\
It is the form of $t$ that shows the causal character of geodesics. Since the expression for $t$ is different in each of the three cases induced by the parameter $a$ and does generally include all parameters from $B$ to $F$, there are too many possibilities of asymptotic behavior to plot examples for all of them. Therefore, we are going to illustrate the behavior of $t$ by some examples focussing on causal geodesics and mention some properties as we go along.\\

\emph{A.} In this case the expression for $t$ is given by \eqref{solt1}. The inequalities
\begin{equation}
 4\beta < 0 < 4\nu < 1 < 4\mu < 2 < 4\alpha
\end{equation}
imply the limit
\begin{equation}
 \lim_{s \searrow -B}t(s) =  \isgn(u_2) \, \infty.
\end{equation}
Using the same inequalities we also get
\begin{equation}
 \lim_{s \to \infty} t(s) = 
  \begin{cases}
  \isgn(u_1)  \quad 	& \text{if} \quad u_1 \neq 0,			 	\\
  \isgn(k) 		& \text{if} \quad u_1 = 0 \text{, } k \neq 0,		\\
  \isgn(u_3) 		& \text{if} \quad u_1=k=0 \text{, } u_3 \neq 0,	\\
  \isgn(G) 		& \text{if} \quad u_1=u_3=k=0 \text{, } G \neq 0,	\\
  \isgn(u_4) \		& \text{if} \quad u_1=u_3=k=G=0.			\\
 \end{cases}
\end{equation}

\emph{B.} Analyzing \eqref{solt2}, we come to the conclusion that the asymptotic behavior of $t$ for $s \searrow -B$ is the same as in case \emph{A}. Furthermore,
\begin{equation}
 \lim_{s \to \infty} t(s) = 
  \begin{cases}
  \isgn(u_1)  \quad 	& \text{if} \quad u_1 \neq 0,			 	\\
  \isgn(k) 		& \text{if} \quad u_1 = 0 \text{, } k \neq 0,		\\
  \isgn(v_2) 		& \text{if} \quad u_1=k=0 \text{, } v_2 \neq 0,		\\
  \isgn(v_1) 		& \text{if} \quad u_1=k=v_2=0 \text{, } v_1 \neq 0,	\\
  \isgn(G) 		& \text{if} \quad u_1=v_2=v_1=k=0.			\\
 \end{cases}
\end{equation}

In the following figure, where we illustrate cases \emph{A} and \emph{B}, we have $u_1 > 0$ and $u_2>0$.\\

\begin{figure}[here] \centering
 \includegraphics[scale=0.5]{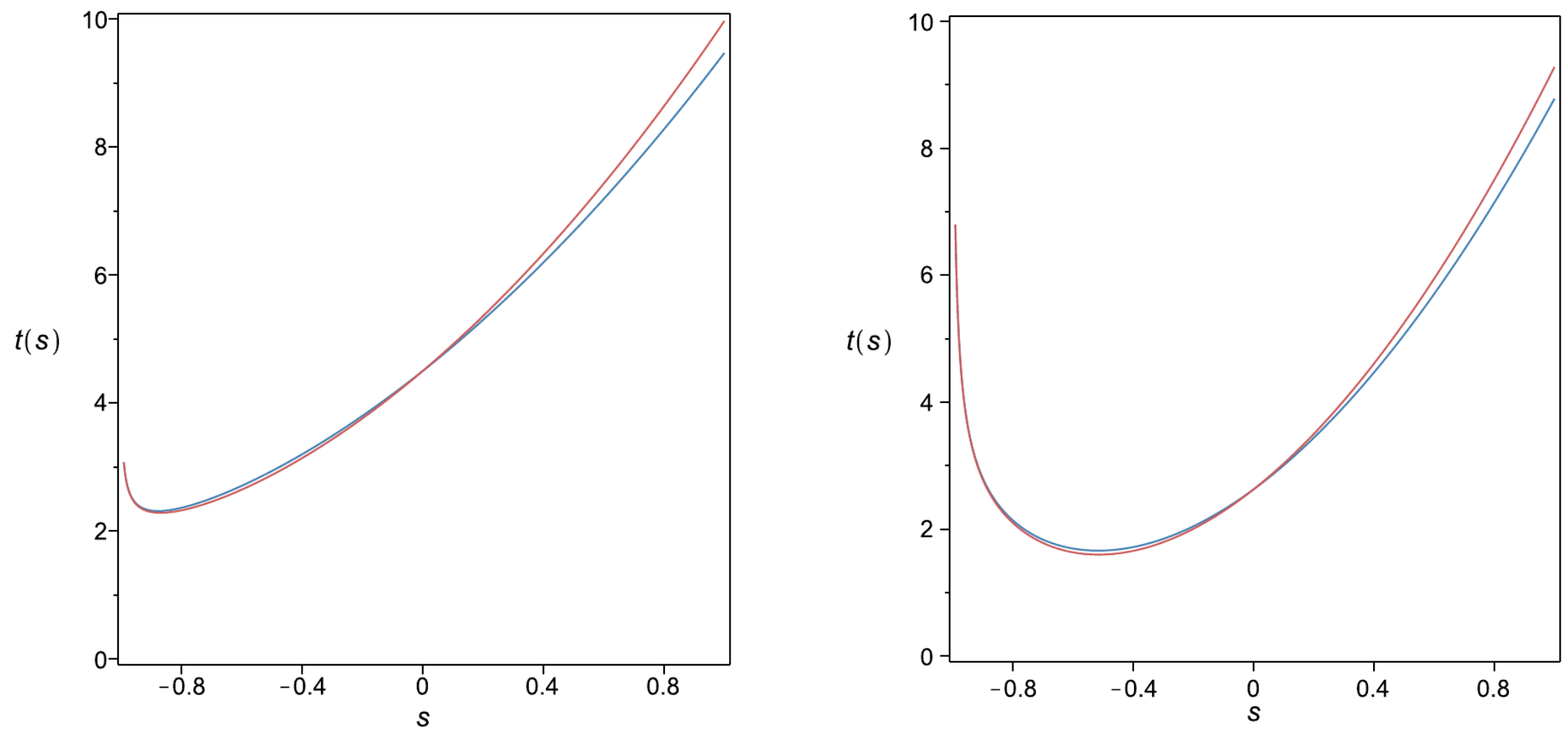}
 \caption{The function $t$ for $a=\frac{1}{16}$ (left) and $a=\frac{1}{8}$ (right). In either case the red curve is a lightlike geodesic, whereas the blue one is timelike. All parameters, except $a$, are equal to $1$ in both plots.}
\end{figure}

\emph{C.} Here we have to deal with \eqref{solt3}. The behavior of $t$ as $s$ tends to $-B$ from above is the same as in the foregoing two cases. This is also true for $\lim_{s \to \infty} t(s)$ except if $u_1=k=0$. In that case, $t$ does in general not converge as $s$ tends to infinity.

\begin{figure}[here] \centering 
 \includegraphics[scale=0.3]{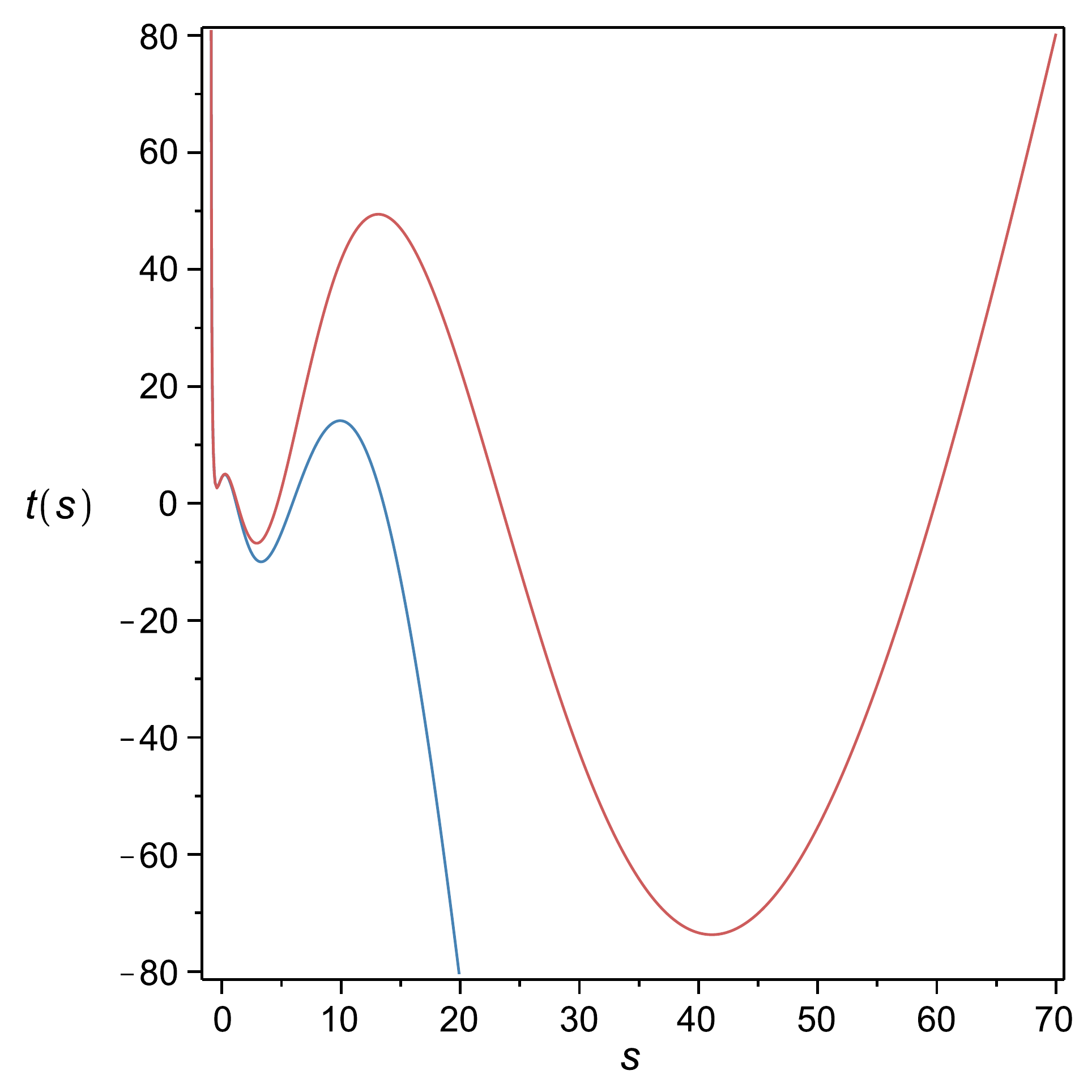}
 \caption{The function $t$ for a lightlike geodesic (red) and a timelike geodesic (blue) in the case of $a>\frac{1}{8}$. Here, all parameters, including $a$, are equal to $1$.} \label{figt3}
\end{figure}

Thus in Fig. \ref{figt3}, the curve in red is the $t$-component of a lightlike geodesic with $u_1 =0$. For the blue $t$-component of the timelike geodesic shown we have $u_1 = 0$ as well, but since here $k=-1$, there is no oscillation and $\lim_{s \to \infty} t(s) = - \, \infty$.

\section{Ori's Spacetime of 1993} \label{chap-Ori93}
Our exposition here is mainly based on \cite{Ori93}, but we also refer to \cite{Ori94} and \cite{Ori96} for further details about the physical relevance of this spacetime.

\subsection{The Manifold}
Topologically the manifold has the form $M:=\R^4$.
In order to specify the metric we start with standard coordinates $(t,x,y,z)$ on $M$ and introduce polar coordiantes on the planes of constant $t$ and $z$ according to the usual formulas 
\begin{align}
\left\{ \begin{aligned}
   x &= r \cos \phi \\  y &= r \sin \phi
\end{aligned} \right.
 \qquad \text{with} \quad r \in (0,\infty), \quad \phi \in (-\pi, \pi).
\end{align}
Interpreting $\phi$ as a circular coordinate, we obtain the alternative coordinate system $(t,r,\phi,z)$ on $M \setminus S$, where $S:=\{ p \in M \colon x=y=0 \}$. Away from $S$ we will use these coordinates from now on unless otherwise mentioned.
The metric then reads
\begin{equation} [g_{ij}]:=
 \begin{bmatrix}
  -1 & 0 & ahrt & 0\\
  0 & 1 & -bhr(r-r_0) & 0\\
  ahrt & -bhr(r-r_0) & r^2 \left( 1+h^2 \left( b^2\rho ^2 - a^2 t^2 \right) \right) & -bhrz 
  0 & 0 & -bhrz & 1
 \end{bmatrix}.
\end{equation}
Here $a,b,r_0 >0$ are constants, $\rho(r,z)^2:=(r-r_0)^2 + z^2$ and $h: \R \to \R$ is a function of class at least $C^2$ depending on $\rho$ satisifying the following conditions: 
\begin{enumerate}
\item $\quad 0 \leq h(\rho) \leq 1$ \quad for all $\rho$,
\item $\quad h(\rho)=1$ \quad for $\rho \leq 0$,
\item $\quad h(\rho) =0$ \quad for $\rho \geq d$,
\item $\quad h'(\rho) < 0$ \quad for $0<\rho<d$.
\end{enumerate}
For the fourth constant $d$ we stipulate $0<d<r_0$. 
The sets $\{ p \in M \colon t(p)=t_0 , \rho(p) = \rho_0 < r_0 \}$ are tori (in particular this is the case for $\rho = d$ ) and when talking about such tori as being of the form $\rho = const$ we tacitly assume $t=const$.
\begin{figure}[here] \centering
 \includegraphics[scale=0.8]{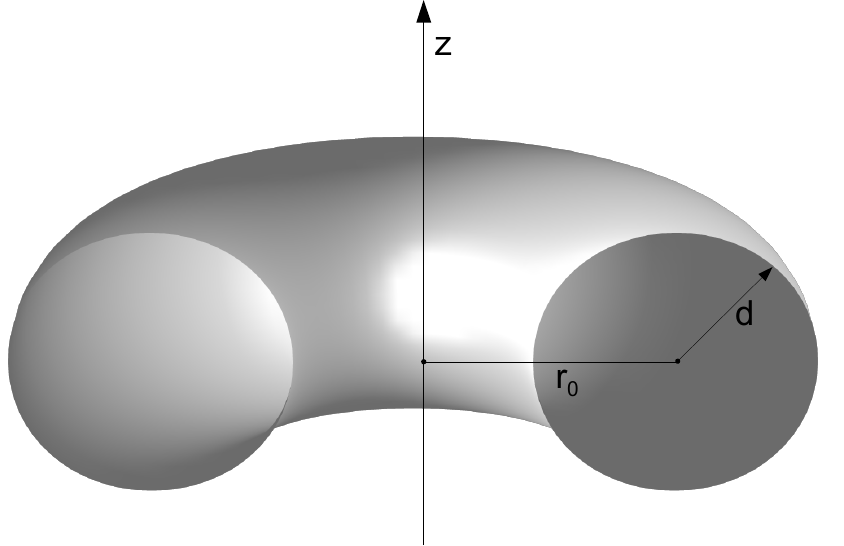}
 \caption{Illustration of the torus $T$. Outside $T$ the metric is equal to the Minkowski metric.}
\label{torus}
\end{figure}
The function $h$ confers a special role to the parameter $d$ because by means of $h$ the metric $g$ reduces to the Minkowski metric 
\begin{equation}
  -dt \otimes dt + dr \otimes dr +r^2 d\phi \otimes d\phi +dz \otimes dz
\end{equation}
outside the torus $T:= \{p \in M \colon \rho(p) =d \}$, which at the same time shows that $g$ is well defined on $M$ since the gap $S$ is filled in by setting $g$ to be the Minkowski metric on $S$, too. The component functions of $g$ are at least $C^2$ everywhere, and we calculate $\det [g_{ij}] = -r^2$. Hence, $g$ is non degenerate. One can easily verify that the vector fields
\begin{equation*}
 E_1 \pdef \partial_t, \quad E_2 \pdef \partial_r, \quad E_3 \pdef aht \partial_t + bh(r-r_0) \partial_r + \frac{1}{r} \partial_\phi + bhz \partial_z, \quad E_4 \pdef \partial_z
\end{equation*}
define an orthonormal frame on the range of the coordinates $(t,r,\phi,z)$, and the relations $\langle E_1,E_1 \rangle = -1$ and $\langle E_k,E_k \rangle = 1$ for $k \neq 1$ show that $g$ has correct Lorentzian signature (throughout $M$).  
Everything presented here is independent of the explicit form of $h$, provided the conditions (i)-(iv) listed above are satisfied. However, there are circumstances when more information about $h$ is needed. Calculating numerical values of certain quantities or the explicit form of geodesics are such examples. 
Algebraically easy to handle is Ori's choice for $h$ which is
\begin{align} h(\rho)=
 \begin{cases}
  1 \qquad 	\qquad	\qquad					& \rho \leq 0, \\
  \left( 1 - \left( \frac{\rho}{d} \right) ^4 \right) ^3	& 0 < \rho <d, \\
  0								& \rho \geq d.
 \end{cases}
\end{align}
In this case one encounters the drawback of h being continuously differentiable up to order 2 only. A different choice which results in a smooth version would be
\begin{align} h(\rho) = 
 \begin{cases}
  1      \qquad 					& \rho \leq 0, \\
  \frac{ \eta(d-\rho) }{ \eta(d-\rho) + \eta(\rho) }	& 0<\rho<d, \\
  0							& \rho \geq d,
 \end{cases}
\qquad \text{where} \quad \eta(t) \pdef
\begin{cases}
 \exp\left(- \frac{1}{t} \right)	 \quad 	& t > 0, \\
 0						& t \leq 0.
\end{cases}
\end{align}
Here $h$ is a typical cutoff function as used in differential geometry.

\subsection{CTCs in $M$}
For any given real number $t_0$ let $H_{t_0}$ denote the hypersurface $t^{-1}(t_0)$ in $M$.
We are interested in characteristics of the set
\begin{equation*}
 M_{CV}(t) := \{ p \in H_t \colon \exists \, \gamma \colon I \to H_t, \, p \in \gamma(I), \, \gamma \text{ closed and timelike} \}, 
\end{equation*}
where the subscript is meant to stand for chronology violation.  
First consider the curves $ \gamma \colon [-\pi,\pi] \to H_t, \, s \mapsto (t,r,s, z)$ for constant values of $t,r$ and $z$ satisfying $\rho(r,z) < d$ (which implies $h(\rho) \neq 0$, a fact we will need later). We have 
\begin{equation} \label{tanvec}
 \langle \gamma' , \gamma' \rangle = \langle \partial_\phi , \partial_\phi \rangle = r^2 ( 1 + h^2 ( b^2\rho^2 - a^2t^2 ) )
\end{equation}
which is independent of the curve parameter $s$. To check for the sign of this expression we need to look at 
\begin{equation} \label{rhot}
 1 + h^2 (b^2 \rho^2 - a^2 t^2 ) = 0 \quad \Leftrightarrow \quad a^2 t^2 = b^2 \rho^2 + h^{-2}.
\end{equation}
The function $ [0,d) \to [1,\infty), \, \rho \mapsto b^2 \rho^2 + h^{-2} $ is strictly increasing (in view of the properties of $h$) and surjective (hence bijective). Therefore, \eqref{rhot} has a unique solution $\rho_t$ iff $|t| \geq \frac{1}{a}$. The function $[\frac{1}{a}, \infty) \to [0,d), \, t \mapsto \rho_t$ is likewise strictly increasing. Given these results we deduce from \eqref{tanvec} and \eqref{rhot} that for $at \geq 1$ (now we index $\gamma$ by two ``parameters" $r,z$)
\begin{align} \gamma_{r,z} \quad \text{is} \quad
 \begin{cases}
  \text{timelike} \qquad		& \text{if} \quad \rho(r,z) < \rho_t, \\
  \text{null}				& \text{if} \quad \rho(r,z) = \rho_t, \\
  \text{spacelike}			& \text{if} \quad \rho(r,z) > \rho_t.
 \end{cases}
\end{align}
Throughout the range $|t|<\frac{1}{a}$ the quantity \eqref{tanvec} above is always positive and $\gamma$ is spacelike, irrespective of the values of $r,z$. In conclusion we know that \[\left\{ p \in H_t \colon \rho(r,z) < \rho_t \right\} \subset M_{CV}(t).\] But we are also able to confine $M_{CV}(t)$ as follows. Calculation gives
\begin{equation} \label{hypsurf}
 g^{tt} = h^2a^2t^2 - 1
\end{equation}
and by Prop.~\ref{prop-hypersurfaces-in-coord}, $H_t$ is a spacelike hypersurface for $|t|<\frac{1}{a}$. The situation changes when $t \geq \frac{1}{a}$. Then we have
\begin{equation} H_t \quad \text{is} \quad
 \begin{cases}
  \text{timelike} \qquad 	& \text{for} \quad	 \rho < \bar{\rho_{t}}, \\
  \text{spacelike} \qquad 	& \text{for} \quad 	 \rho > \bar{\rho_{t}},
 \end{cases}
\end{equation}
where $\bar{\rho_{t}}$ is determined by $h(\bar{\rho_t})=\frac{1}{at}$, i.e., $\bar{\rho_{t}}$ is the (unique) zero of \eqref{hypsurf}.  Comparing the equations for $\rho_t$ and $\bar{\rho_{t}}$ yields the first part of the following inequality and the second part results from the properties of $h$:
\begin{equation}
 \rho_t \leq \bar{\rho_{t}} < d.
\end{equation}
In any given hypersurface $H_t$ the region $\rho > \bar{\rho_{t}}$ is thus free of timelike curves which implies $M_{CV}(t) \subset \{p \in M \colon \rho(p) \leq \bar{\rho_t} \}$. In particular, all CTCs in $H_t$, including the ones in the torus $\rho = \rho_t$, are contained in the torus $\rho=\bar{\rho_{t}}$. This situation is illustrated in Figure \ref{grtorus}.
\begin{figure}[here] 
\includegraphics{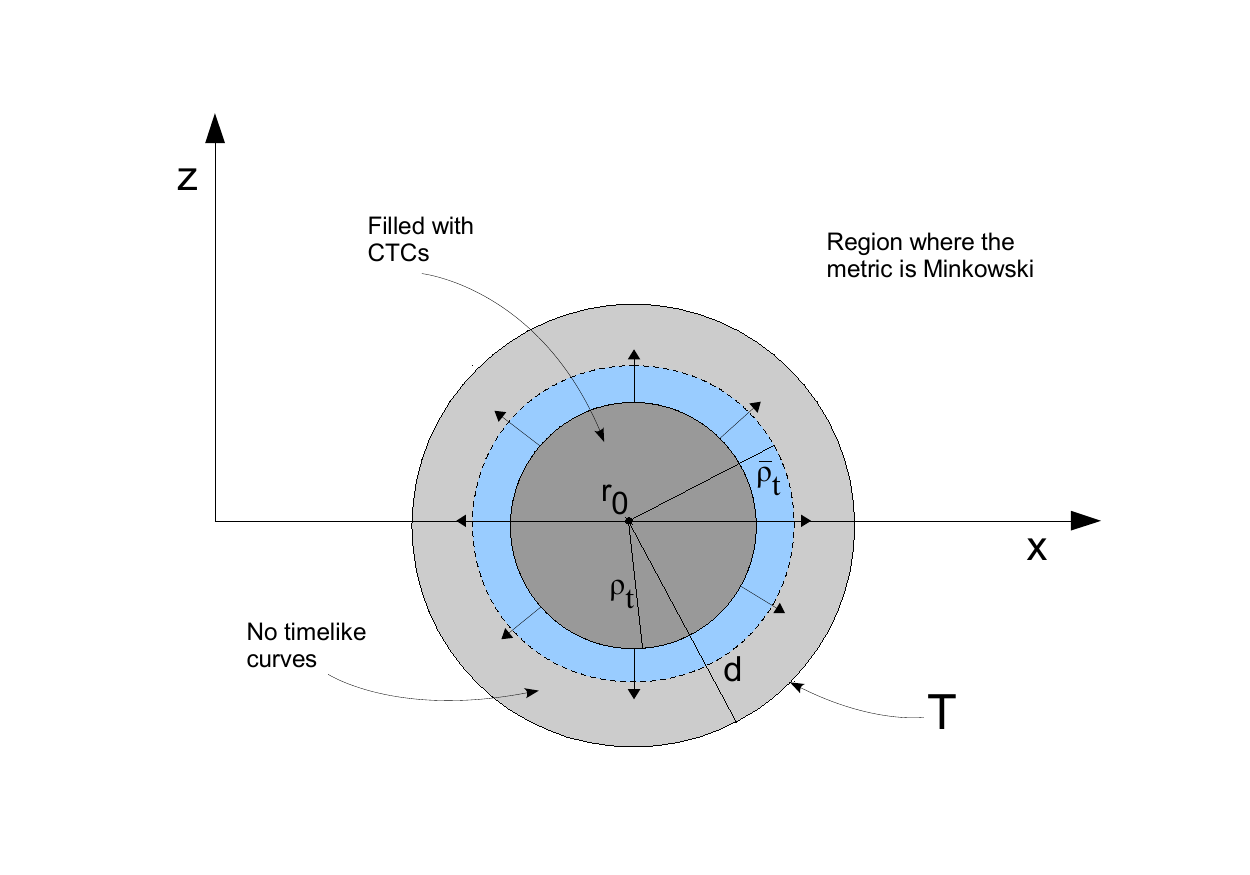}
\caption{For fixed $t>\frac{1}{a}$ this is a section of the torus $T$ containing the $x$ and $z$-axis (the $t$ and $y$-coordinates are suppressed). The growth of the two smaller tori $\rho=\rho_t$ and $\rho=\bar{\rho_t}$ with $t$ is indicated by the outward pointing arrows.}
\label{grtorus}
\end{figure}

Note that for all $t$, viewed as subsets of $H_t$, the set $M_{CV}(t)$ is contained in the compact set $T$.

\subsection{Tipping of Lightcones}
\begin{figure}[here]
 \centering
 \includegraphics[scale=0.61]{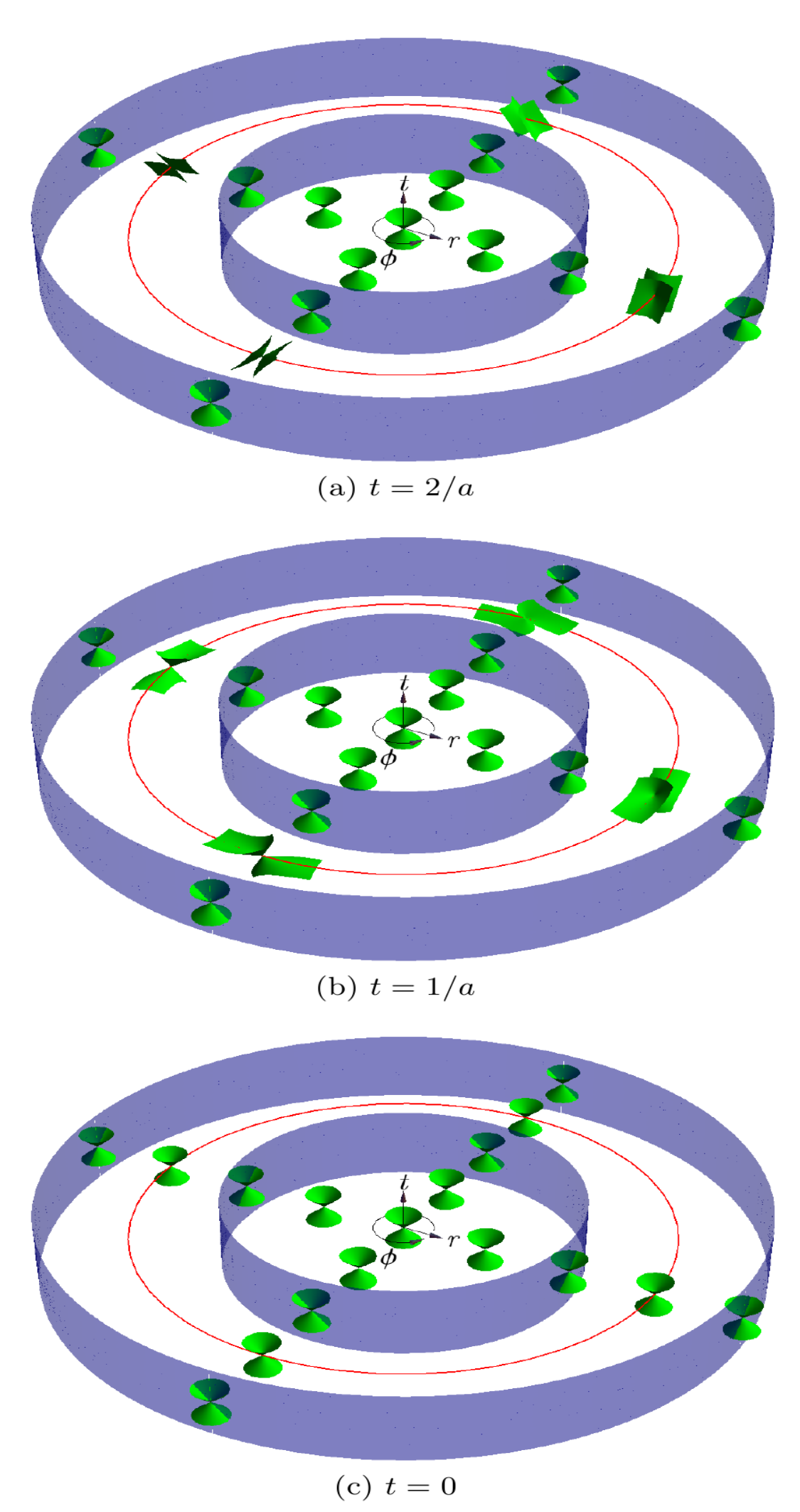}
 \caption{Inside the torus $T$ lightcones tip over as $t$ increases. T is represented by the blue stripes and the red curve is given by $\rho = 0$. The latter becomes a closed null curve at $t=\frac{1}{a}$ and is a CTC for $t>\frac{1}{a}$. This illustration is due to T. Sch\"onfeld \cite{Schoenfeld}.}
 \label{tipping}
\end{figure}
The fact that the curves $\gamma_{r,z}$ from above are CTCs filling a torus of growing radius $\rho_t$, implies that the tangent vector $\gamma_{r,z} ' = \partial_\phi$ becomes timelike at points inside the torus $T$, where it was null or spacelike before (i.e., for smaller values of $t$). This is reflected by the following argument, where we place ourselves at $r=r_0, z=0$ and consider a tangent vector of the form $X=\alpha \partial_\phi + \beta \partial_z$ with real coefficients $\alpha$, $\beta$. We have
\begin{equation}
 \langle X,X \rangle = \beta^2 + \alpha^2 r_0^2 (1-a^2t^2),
\end{equation}
which is positive for small values of $t$, zero at a certain $t_0$ and negative for $t$ large enough. Correspondingly, the causal character of $X$ is spacelike first, lightlike at $t_0$ and then timelike. Using the concept of lightcones, $X$ is initially (i.e., $t$ small) situated outside the lightcone at the corresponding point of $M$, at $t_0$ it lies on the lightcone and can finally be found inside the lightcone when $t$ has grown sufficiently. This process of tipping of lightcones shows geometrically how CTCs emerge and is depicted in figure \ref{tipping} (from \cite{Schoenfeld}), where we can see the set $\{ (t,r,\phi,z) \in M \colon z=0 \}$. The coordinate $t$ increases in the vertical direction of the picture and for the values $t= 0, \frac{1}{a}, \frac{2}{a}$ the same set of lightcones in the planes spanned by $r$ and $\phi$ is plotted.

\end{document}